\documentclass[a4paper, 11pt]{article}

\usepackage{amsmath}
\usepackage{amssymb}
\usepackage{amsthm}
\usepackage{amsfonts}
\usepackage{amstext}
\usepackage{amsopn}
\usepackage{amsxtra}
\usepackage[latin1]{inputenc}
\usepackage{dsfont}
\usepackage{bm}

\newtheorem{thm}{Theorem}
\newtheorem{lemma}{Lemma}

\newtheorem{prop}{Proposition}
\theoremstyle{definition}

\newtheorem{remark}{Remark}

\newcommand\R{{\ensuremath {\mathbb R} }}
\newcommand\C{{\ensuremath {\mathbb C} }}

\newcommand\1{{\ensuremath {\mathds 1} }}

\renewcommand\phi{\varphi}

\newcommand{\alp}{\bm{\alpha}}

\newcommand{\gH}{\mathfrak{H}}
\newcommand{\gS}{\mathfrak{S}}

\newcommand{\cP}{\mathcal{P}}

\newcommand\ii{{\ensuremath {\infty}}}

\newcommand{\norm}[1]{ \left| \! \left| #1 \right| \! \right| }

\def\tr{\mathop{\rm tr}\nolimits} 

\newcommand{\cF}{\mathcal{F}}
\newcommand{\F}{\mathcal{F}}
\newcommand{\cT}{\mathcal{T}}
\newcommand{\cK}{\mathcal{K}}
\newcommand{\cB}{\mathcal{B}}
\newcommand{\cC}{\mathcal{C}}
\newcommand{\cA}{\mathcal{A}}

\newcommand{\half}{\mbox{$\frac12$}}

\begin{document}
\title{A Nonlinear Model for Relativistic Electrons At Positive Temperature}
\author{C. HAINZL, M. LEWIN and R. SEIRINGER}
\
\bigskip\bigskip

\begin{center}
  \Large \textbf{A Nonlinear Model for Relativistic Electrons at Positive Temperature}
\end{center}

\medskip

\begin{center}
 \large Christian HAINZL$^a$, Mathieu LEWIN$^b$ \& Robert SEIRINGER$^c$
\end{center}

\medskip

\begin{center}
\small

 $^a$Department of Mathematics, UAB, Birmingham, AL 35294-1170, USA.

\texttt{hainzl@math.uab.edu}

\smallskip

$^b$CNRS \& Department of Mathematics (CNRS UMR8088), University of Cergy-Pontoise, 2 avenue Adolphe Chauvin, 95302 Cergy-Pontoise Cedex, FRANCE.

\texttt{Mathieu.Lewin@math.cnrs.fr}

\smallskip

$^c$Department of Physics, Jadwin Hall, Princeton University, P.O. Box 708, Princeton, New Jersey 08544, USA.

\texttt{rseiring@math.princeton.edu}
\end{center}

\begin{center}
 \it February 26, 2008
\end{center}

\smallskip

\begin{abstract}
  We study the relativistic electron-positron field at positive
  temperature in the Hartree-Fock-approximation. We consider both the
  case with and without exchange term, and investigate the existence
  and properties of minimizers. Our approach is non-perturbative in
  the sense that the relevant electron subspace is determined in a
  self-consistent way. The present work is an extension of previous
  work by Hainzl, Lewin, S\'er\'e, and Solovej where the case of zero
  temperature was considered.
\end{abstract}

\tableofcontents

\section*{Introduction}\addcontentsline{toc}{section}{Introduction}

In Coulomb gauge and when photons are neglected, the Hamiltonian of
Quantum Electrodynamics (QED) reads formally \cite{Hei,BD,HLSo,HLSS}
\begin{equation}
\label{Hamiltonian01} \mathbb{H}^\phi = \int \Psi^*(x) D^0 \Psi(x)
\,dx -\int\phi(x){\rho}(x)dx+ \frac{\alpha}2 \iint
\frac{\rho(x)\rho(y)}{|x-y|}dx\,dy\,.
\end{equation}
Here $\Psi(x)$ is the second-quantized field operator satisfying
the usual anti-commutation relations, and $\rho(x)$ is the density
operator
\begin{equation}
\label{Rho1}
\rho(x)=\frac{\sum_{\sigma=1}^4[\Psi^*(x)_\sigma,\Psi(x)_\sigma]}{2}=
\frac{\sum_{\sigma=1}^4\left\{\Psi^*(x)_\sigma\Psi(x)_\sigma-\Psi(x)_\sigma\Psi^*(x)_\sigma\right\}}{2},
\end{equation}
where $\sigma$ is the spin variable. In \eqref{Hamiltonian01},
$D^0=-i\alp\cdot\nabla+\beta$ is the usual free Dirac operator,
$\alpha$ is the \emph{bare} Sommerfeld fine structure constant and
$\phi$ is the external potential. The matrices $\alp =
(\alpha_1,\alpha_2,\alpha_3)$ and $\beta$ are the usual $4\times 4$
anti-commuting Dirac matrices. We have chosen a system of units
such that $\hbar=c=m=1$. In QED, one main issue is the minimization of
the Hamiltonian \eqref{Hamiltonian01}. However, even if we implement
an UV-cutoff, the Hamiltonian is unbounded from below, since the
particle number can be arbitrary.

In a formal sense this problem was first overcome by Dirac, who
suggested that the vacuum is filled with infinitely many particles
occupying the negative energy states of the free Dirac operator
$D^0$. With this axiom Dirac was able to conjecture the
existence of holes in the Dirac sea which he interpreted as {\em
  anti-electrons} or {\em positrons}. His prediction was verified by
Anderson in 1932.  Dirac also predicted \cite{D1,D2} the phenomenon of
vacuum polarization: in the presence of an electric field, the virtual
electrons are displaced and the vacuum acquires a non-uniform charge density.

In Quantum Electrodynamics Dirac's assumption is sometimes implemented
via {\em normal ordering} which essentially consists of subtracting
the kinetic energy of the negative free Dirac sea, in such a way that
the kinetic energy of electrons as well as positrons (holes) becomes
positive. With this procedure the distinction between electrons and
positrons is put in by hand.

It was pointed out in \cite{HLSo} (see also the review \cite{HLSS}),
however, that normal ordering is probably not well suited to the case
$\alpha\neq0$ of interacting particles (the interaction is the last
term of \eqref{Hamiltonian01}). Instead a procedure was presented
where the distinction between electrons and positrons is not an input
but rather a consequence of the theory. The approach of \cite{HLSo} is
rigorous and fully non-perturbative, but so far it was only applied to
the mean-field (Hartree-Fock) approximation, with the photon field
neglected. It allowed to justify the use of the Bogoliubov-Dirac-Fock
model (BDF) \cite{CI}, studied previously in \cite{HLS1,HLS2,HLS3}.
The purpose of the present paper is to extend
these results to the nonzero temperature case.

The methodology of \cite{HLSo} is a two steps procedure. First, the
free vacuum is constructed by minimizing the Hamiltonian
\eqref{Hamiltonian01} over Hartree-Fock states in a box with an
ultraviolet cut-off, and then taking the thermodynamic limit when
the size of the box goes to infinity. The limit is a Hartree-Fock
state $\cP^0_-$ describing the (Hartree-Fock) free vacuum
\cite{HLSo,HLSS}. It has an infinite energy, since it contains
infinitely many virtual particles forming the (self-consistent)
Dirac sea. We remark that this state is not the usual sea of
negative electrons of the free Dirac operator because all
interactions between particles are taken into account, but it
corresponds to filling negative energies of an effective mean-field
translation invariant operator.

The second step of \cite{HLSo} consists of constructing an energy
functional that is bounded from below  in the presence of an
external field, by subtracting the (infinite) energy of the free
self-consistent Dirac sea.  The key observation is that the
difference of the energy of a general state $P$ minus the (infinite)
energy of the free vacuum $\cP^0_-$ can be represented by an
effective functional (called Bogoliubov-Dirac-Fock (BDF) \cite{CI})
which only depends on $Q=P-\cP^0_-$, describing the variations with
respect to the free Dirac sea.  The BDF energy was studied in
\cite{HLS1,HLS2,HLS3}. The existence of ground states was shown for
the vacuum case in \cite{HLS1,HLS2} and in charge sectors in
\cite{HLS3}. For a detailed review of all these results, we refer to
\cite{HLSS}. An associated time-dependent evolution equation, which
is in the spirit of Dirac's original paper \cite{D1}, was studied in
\cite{HLSp}.

\bigskip

Let us now turn to the case of a non zero temperature
$T=1/\beta>0$. We consider a Hartree-Fock state with one-particle density matrix
$0\leq P\leq 1$.
Because
of the definition of the Hamiltonian \eqref{Hamiltonian01} and the
anticommutator in \eqref{Rho1}, it is more convenient to consider as
variable the \emph{renormalized density matrix} $\gamma=P-1/2$. We
remark that the anticommutator in \eqref{Rho1} is a kind of
renormalization which does not depend on any reference as normal
ordering does (it just corresponds to subtracting the identity divided
by 2). The anticommutator of \eqref{Rho1} is due to Heisenberg
\cite{Hei} (see also \cite[Eq. $(96)$]{Pauli}) and it is necessary for
a covariant formulation of QED, see \cite[Eq. $(1.14)$]{Sch1} and
\cite[Eq. $(38)$]{Dy1}.

Computing the free energy of our Hartree-Fock state using \eqref{Hamiltonian01} (and ignoring infinite constant terms) one arrives at the following free energy functional
 \cite{HLSo,BLS}
\begin{multline}\label{energyQED}
\mathcal{F}^{\rm QED}_T (\gamma) =\tr(D^0 \gamma) -\alpha \int
\varphi(x) \rho_\gamma(x) + \frac \alpha 2 \iint \frac{\rho_\gamma(x)
\rho_\gamma(y)}{|x-  y|}\\ - \frac{\alpha}{2} \int \int \frac{\tr_{\C^4} |\gamma(x
,y)|^2}{|x-y|} -TS(\gamma)
\end{multline}
where the \emph{entropy} is given by the formula
\begin{equation}
S(\gamma)=-\tr\left((\half+\gamma)\ln(\half+\gamma)\right)-\tr\left((\half-\gamma)\ln(\half-\gamma)\right).
\label{def_entropy}
\end{equation}
The (matrix-valued) function $\gamma(x,y)$ is the formal {integral kernel} of the operator $\gamma$ and $\rho_\gamma(x):=\tr_{\C^4}\gamma(x,x)$ is the associated charge density. The above formulas are purely formal; they only make sense in a finite box with an ultraviolet cut-off, in general.

As in \cite{HLSo} the first step is to define the free vacuum at
temperature $T$, which is the formal minimizer of \eqref{energyQED}
when $\phi=0$. Following \cite{HLSo}, one can first confine the system
to a box, then study the limit as the size of the box goes to infinity
and identify the free vacuum as the limit of the sequence of ground
states. Alternatively, it was proved in \cite{HLSo} that the free
vacuum can also be obtained as the unique minimizer of the free energy
per unit volume. In the nonzero temperature case, this energy reads
\begin{align*}
&\cT_T(\gamma)=\\
&\frac{1}{(2\pi)^{3}}\int_{B(0,\Lambda)}\tr_{\C^4}[
D^0(p)\gamma(p)]dp -\frac{\alpha}{(2\pi)^5}\iint_{B(0,\Lambda)^2}\frac{\tr_{\C^4}
[\gamma(p)\gamma(q)]}{|p-q|^2}dp\,dq\\
&+\frac{T}{(2\pi)^3}\int_{B(0,\Lambda)}\!\!\tr_{\C^4}\!\!\big[
\left(\half+\gamma(p)\right)\ln\left(\half+\gamma(p)\right)+\left(\half-\gamma(p)\right)\ln\left(\half-\gamma(p)\right)\!\big]dp
\end{align*}
and it is defined for translation-invariant states $\gamma=\gamma(p)$
only, under the constraint $-1/2\leq\gamma\leq 1/2$. Here,
$B(0,\Lambda)$ denotes the ball of radius $\Lambda$ centered at the
origin. The real number $\Lambda>0$ is the ultraviolet cut-off. We
shall prove in Theorem \ref{thm_free_case} that the above energy has a
unique minimizer $\tilde\gamma^0$, and prove several interesting
properties of it. In particular, we shall see that it satisfies a
nonlinear equation of the form
\begin{equation}
\tilde\gamma^0=\frac12\left( \frac{1}{1+e^{ \beta D_{\tilde\gamma^0}}}-\frac{1}{1+e^{-\beta D_{\tilde\gamma^0}}}\right) \\
\end{equation}
i.e. it is the Fermi-Dirac distribution of a (self-consistent) free Dirac operator, defined as
$$D_{\tilde\gamma^0}=D^0-\alpha\frac{\tilde\gamma^0(x,y)}{|x-y|}.$$
(The last term stands for the operator having this integral kernel.)
This extends results of \cite{HLSo} to the $T>0$ case.

The next step is to formally subtract the (infinite) energy of $\tilde{\gamma}^0$ from the energy of any state $\gamma$. In this way one obtains a Bogoliubov-Dirac-Fock free energy at temperature $T=1/\beta$ which can be formally written as
\begin{align}
\mathcal{F}_T (\gamma) &=\text{``}\mathcal{F}^{\rm QED}_T (\gamma)-\mathcal{F}^{\rm QED}_T (\tilde\gamma^0)\text{''}\nonumber\\
&=TH(\gamma,\tilde{\gamma}^0) -\alpha \int
\varphi(x) \rho_{[\gamma-\tilde{\gamma}^0]}(x) + \frac \alpha 2 \iint \frac{\rho_{[\gamma-\tilde{\gamma}^0]}(x)
\rho_{[\gamma-\tilde{\gamma}^0]}(y)}{|x-  y|}\nonumber\\ &\qquad\qquad\qquad- \frac{\alpha}{2} \int \int \frac{\tr_{\C^4}|(\gamma-\tilde{\gamma}^0)(x
,y)|^2}{|x-y|}\label{energyBDF_intro}
\end{align}
where $H$ is the \emph{relative entropy} formally defined as
$$TH(\gamma,\tilde{\gamma}^0)=\text{``}\tr(D_{\tilde{\gamma}^0}(\gamma-\tilde{\gamma}^0))-TS(\gamma)+TS(\tilde\gamma^0)\text{''}.$$
We shall consider external field of the form $\varphi =
\nu\ast \frac 1{|x|}$, where $\nu$ represents the density
distribution of the external particles, like nuclei, or molecules.

In Section \ref{sec:external}, we show how to give a correct
mathematical meaning to the previous formulas and we prove that the
BDF free energy is bounded from below.  An important tool is the
following inequality
\begin{equation}\label{freebdfenergy_intro}
TH(\gamma,\tilde \gamma^0) \geq \tr\left[ |D_{\tilde
\gamma^0}|(\gamma - \tilde \gamma^0)^2\right] \geq \tr \left[
|D^0|(\gamma - \tilde \gamma^0)^2\right].
\end{equation}
This implies that the relative entropy can control the exchange term
and enables us to show that $\F_T$ is bounded from below.

Unfortunately, like for the $T=0$ case, the free BDF energy is not
convex, which makes it a difficult task to prove the existence of a
minimizer. Although we leave this question open, we derive some
properties for a potential minimizer in Section \ref{sec:external}. In
particular we prove that any minimizer $\gamma$ satisfies the
following nonlinear equation
\begin{equation}
\gamma=\frac12\left( \frac{1}{1+e^{ \beta D_{\gamma}}}-\frac{1}{1+e^{-\beta D_{\gamma}}}\right)
\end{equation}
where the (self-consistent) Dirac operator reads
$$D_{\gamma}=D^0+\alpha\rho_\gamma\ast|\cdot|^{-1}-\alpha\phi-\alpha\frac{\gamma^0(x,y)}{|x-y|}.$$
Compared with the zero temperature
case, the main difficulty in proving the existence of a minimizer
comes from localization issues of the relative entropy which are more
involved than in the zero temperature case.

\bigskip

As a slight simplification, we thoroughly study the \emph{reduced}
Hartree Fock case for $T > 0$, where the exchange term (the first term
of the second line of \eqref{energyQED}) is neglected. In the
zero-temperature case, this model was already studied in detail in
\cite{HLS2} and \cite{GLS}. The corresponding free vacuum is now
simple: it is the Fermi-Dirac distribution corresponding to the usual
free Dirac operator $D^0$,
$$\gamma^0=\frac12\left(\frac1{1+e^{\beta D^0}}- \frac1{1+e^{-\beta D^0}}\right).
 $$
 The {\em reduced Bogoliubov-Dirac-Fock} free energy is obtained in the same way as before by subtracting the infinite energy of the free Dirac see $\gamma^0$ to the (reduced) Hartree-Fock energy. It is given by
$$\F^{\rm red}_T(\gamma) = T H(\gamma, \gamma^0)   - \alpha \int \varphi\rho_{[\gamma - \gamma^0]}+ \frac \alpha 2 \iint \frac{\rho_{[\gamma-{\gamma}^0]}(x)
\rho_{[\gamma-{\gamma}^0]}(y)}{|x-  y|}dx\,dy,$$
$H(\gamma,\gamma^0)$ being defined similarly as before.
As this functional is now convex, we can prove in Theorem \ref{thm_exists_rBDF} that it has a unique minimizer $\bar\gamma$, which satisfies the self-consistent equation
$$ \bar\gamma = \frac12\left(\frac1{1+e^{\beta D_{\bar\gamma}}}-\frac{1}{1+e^{-\beta D_{\bar\gamma}}} \right)$$
where
$$D_{\bar\gamma} :=D^0+\alpha\rho_{\bar\gamma-\gamma^0}\ast|\cdot|^{-1}-\alpha\phi $$
in this case.

Additionally we show in Theorem \ref{thm_Debye_rBDF} that this
minimizer has two interesting properties. First, $\bar\gamma-\gamma^0$
is a trace-class operator. In the zero
temperature case, on the other hand, it was proved in \cite{HLS2} that the
minimizer is never trace-class for $\alpha>0$. This was indeed the
source of complications concerning the definition of the trace (and
hence of the charge) of Hartree-Fock states \cite{HLS1} when
$T=0$. This is related to the issue of renormalization
\cite{HLS2,HLSo,GLS}. Although we do not minimize in the trace-class
in the case $T\neq0$ but rather in the Hilbert-Schmidt class because
the free energy is only coercive for the Hilbert-Schmidt norm, it
turns out that the minimizer is trace-class nevertheless.

The second (and related) interesting property shown in Theorem
\ref{thm_Debye_rBDF} below is that the total electrostatic potential
created by the density $\nu$ and the polarized Dirac sea decays very
fast. More precisely we prove that
$$\rho_{\bar \gamma}-\nu\in
L^1(\R^3)\quad\text{and}\quad (\rho_{\bar \gamma} - \nu)\ast \frac 1{|x|} \in
L^1(\R^3).
$$
Necessarily, the charge of $\rho_{\bar \gamma}$ and the charge of the
external sources have to be equal. More precisely the {\em effective}
potential has a much faster decay at infinity than $1/|x|$, which
shows that the {\em effective} potential is screened.  In other words
due to the positive temperature, the particles occupying the Dirac-sea
have enough freedom to rearrange in such a way that the external
sources are totally shielded. Within non-relativistic fermionic plasma
this effect is known as {\em Debye-screening}. Let us emphasize
that in order to recover such a screening, it is essential to
calculate the Gibbs-state in a self-consistent way.

These two properties of the minimizer of the reduced theory probably
also hold for the full BDF model with exchange term. However, like for
the case $T=0$, the generalization does not seem to be
straightforward.

\bigskip

The paper is organized as follows. The first section is devoted to the
presentation of our results for the \emph{reduced} model which is
simpler and for which we can prove much more than for the general
case. In the second section, we consider the original Hartree-Fock
model with exchange term. We prove the existence and uniqueness of the
free Hartree-Fock vacuum, define the BDF free energy in the presence
of an external field and provide some interesting properties of
potential minimizers. In the last section we provide some details of
proofs which are a too lengthy to be put in the main text.

\bigskip

\noindent\textbf{Acknowledgments.} M.L. acknowledges support from the
ANR project ``ACCQUAREL'' of the French ministry of research. R.S. was
partially supported by U.S. NSF grant PHY-0652356. and by an
A.P. Sloan fellowship.

\section{The reduced Bogoliubov-Dirac-Fock free energy }
\subsection{Relative entropy}\label{sec:def_H}
Throughout this paper, we shall denote by $\gS_p(\gH)$ the usual Schatten class of operators $Q$ acting on a Hilbert space $\gH$ and such that $\tr(|Q|^p)<\ii$. The UV cut-off is implemented like in \cite{HLS1,HLS2,HLS3,HLSo} in Fourier space by considering the
Hilbert space
\begin{equation}
\gH_\Lambda:=\left\{\psi \in L^2(\R^3, \C^4)\ |\ {\rm supp}\, \hat \psi \subset
B(0,\Lambda)\right\}\,,
\label{cutoff2}
\end{equation}
with $B(0,\Lambda)$ denoting the ball of radius $\Lambda$ centered at the origin.
We denote by `$\tr$' the usual trace functional on $\gS_1(\gH_\Lambda)$.
Within the reduced theory, the free vacuum at temperature $T=\beta^{-1}>0$ is the self-adjoint operator acting on $\gH_\Lambda$ defined by
\begin{equation}
\gamma^0=\frac12\left(\frac1{1+e^{\beta D^0}}- \frac1{1+e^{-\beta
D^0}}\right). \label{def_free_red}
\end{equation}
Notice when $T\to0$ ($\beta\to\ii$), we recover the usual formula \cite{HLS1,HLS2,GLS}
$\gamma^0=-D^0/2|D^0|$.

We assume that $T>0$ henceforth. Notice that thanks to the cut-off in
Fourier space and the gap in the spectrum of $D^0$, the spectrum of
$\gamma^0$ does not include $0$ or $\pm 1/2$. In fact, it is given by
\begin{multline}
\sigma(\gamma^0) = \left[-\frac12+\frac{e^{-\beta E(\Lambda)}}{1+e^{-\beta E(\Lambda)}} , -\frac12+\frac{e^{-\beta}}{1+e^{-\beta}}\right]\\
\cup \left[\frac12-\frac{e^{-\beta}}{1+e^{-\beta}} , \frac12-\frac{e^{-\beta E(\Lambda)}}{1+e^{-\beta E(\Lambda)}}\right]
\label{spectrum_free}
\end{multline}
where $E(\Lambda)=\sqrt{1+\Lambda^2}$. Also the charge density of the free vacuum $\gamma^0$ at temperature $T$ vanishes:
\begin{equation}
\rho_{\gamma^0}=\frac{1}{2(2\pi)^{3}}\int_{B(0,\Lambda)}\tr_{\C^4}\left(\frac1{1+e^{\beta D^0(k)}}- \frac1{1+e^{-\beta D^0(k)}}\right)dk = 0.
\label{rho_vanishes}
\end{equation}

We shall denote the class of Hilbert-Schmidt perturbations of $\gamma^0$ by $\cK$:
\begin{equation}
\cK:=\left\{\gamma\in\cB(\gH_\Lambda)\ |\ \gamma^*=\gamma,\ -\frac{1}{2}\leq\gamma\leq \frac{1}{2},\ \gamma-\gamma^0\in\gS_2(\gH_\Lambda)\right\}.
\label{def_convex_set}
\end{equation}
The \emph{relative entropy} reads
\begin{multline}
H(\gamma,\gamma^0)=\tr\bigg[\left(\half+\gamma\right)\left(\ln\left(\half+\gamma\right)-\ln\left(\half+\gamma^0\right)\right)\\
+\left(\half-\gamma\right)\left(\ln\left(\half-\gamma\right)-\ln\left(\half-\gamma^0\right)\right)\bigg].
\label{def_H_log}
\end{multline}
Note that since $\gamma\in\cK$ is a compact perturbation of $\gamma^0$, we always have $\sigma_{\rm ess}(\gamma)=\sigma_{\rm ess}(\gamma^0)$. Hence $\sigma(\gamma)$ only contains eigenvalues of finite multiplicity in the neighborhood of $\pm1/2$.
Using the integral formula
\begin{equation}
\ln a - \ln b = -\int_0^\infty \left[ \frac 1{a+t} - \frac 1{b+t}\right]dt=
\int_0^\infty \frac 1{a+t}(a-b)\frac 1{b+t}dt,
\label{formula_log}
\end{equation}
we easily see that Eq.~\eqref{def_H_log} is well defined as soon as $\gamma\in\cK$, $\gamma-\gamma^0\in\gS_1(\gH_\Lambda)$, since the spectrum of $\gamma^{0}$ does not contain $\pm1/2$.

When $\gamma-\gamma^0\in\cK$ is merely Hilbert-Schmidt, we may define the relative entropy by the integral formula
\begin{equation}
\boxed{\ H(\gamma,\gamma^0)=\tr\left(\int_{-1}^1\frac{2}{1+2u\gamma^0}(\gamma-\gamma^0)\frac{1-|u|}{1+2u\gamma}(\gamma-\gamma^0)\frac{1}{1+2u\gamma^0}du\right) \ }
\label{def_H_int}
\end{equation}
It is clear that this provides a well defined object in $\cK$ as one has
$$\forall\gamma\in\cK,\ \forall u\in[-1,1],\quad 0\leq \frac{1-|u|}{1+2u\gamma}\leq 1\quad\text{and}\quad 0\leq\frac{1}{1+2u\gamma^0}\leq \frac{1}{\epsilon}$$
for some $\epsilon>0$, by \eqref{spectrum_free}. It is not difficult
to see that (\ref{def_H_int}) and (\ref{def_H_log}) coincide when
$\gamma-\gamma^0 \in \gS_1(\gH_\Lambda)$. We shall discuss this in the
appendix. But (\ref{def_H_int}) has the advantage of being
well-defined for all $\gamma\in \cK$, and hence we use
(\ref{def_H_int}) for a {\it definition} of $H$ henceforth.

Our first result is the

\begin{thm}[{\bf Properties of Relative Entropy}]\label{prop_rel_entropy}
  The functional $\gamma\mapsto H(\gamma,\gamma^0)$ defined in
  (\ref{def_H_int}) is strongly continuous on $\cK$ for the topology
  of $\gS_2(\gH_\Lambda)$. It is convex, hence weakly lower
  semi-continuous (wlsc).  Moreover, it is coercive on $\cK$ for the
  Hilbert-Schmidt norm:
\begin{equation}
\forall\gamma\in\cK,\qquad TH(\gamma,\gamma^0)\geq \tr\left(|D^0|(\gamma-\gamma^0)^2\right)
\label{estim_below_H}
\end{equation}
where we recall that $T=\beta^{-1}$ is the temperature.
\end{thm}

Coercive in this context means that
$H(\gamma,\gamma^0) \to \infty$ if
$\norm{\gamma-\gamma^0}_{\gS_2(\gH_\Lambda)} \to \infty$. This follows from (\ref{estim_below_H}) since $|D^0|\geq 1$.

\begin{proof}[Proof of Theorem \ref{prop_rel_entropy}]
First we prove that $H(\cdot,\gamma^0)$ is strongly continuous for the $\gS_2(\gH_\Lambda)$ topology. This is indeed a consequence of the following
\begin{lemma}\label{lem:upper_bound_HS}
Let $\gamma,\gamma'\in\cK$. Then we have for some constant $C$ (depending on $\Lambda$) and all $0\leq \eta\leq1$,
\begin{multline}
\left|H(\gamma,\gamma^0)-H(\gamma',\gamma^0) \right|\leq \frac{C}{\eta}\norm{\gamma-\gamma'}_{\gS_2(\gH_\Lambda)}\\+C\eta\left(\norm{\gamma-\gamma^0}_{\gS_2(\gH_\Lambda)}^2+\norm{\gamma'-\gamma^0}_{\gS_2(\gH_\Lambda)}^2\right).
\label{upper_bound_HS}
\end{multline}
\end{lemma}
\begin{proof}
We use Formula \eqref{def_H_int} and split the integrals as follows: $$\int_{-1}^1=\int_{-1}^{-1+\eta}+\int_{-1+\eta}^{1-\eta}+\int_{1-\eta}^{1}.$$
We estimate
\begin{multline*}
\Bigg| \tr\Bigg(\int_{-1+\eta}^{1-\eta}du\bigg(\frac{1}{1+2u\gamma^0}(\gamma-\gamma^0)\frac{1-|u|}{1+2u\gamma}(\gamma-\gamma^0)\frac{1}{1+2u\gamma^0}\\
-\frac{1}{1+2u\gamma^0}(\gamma'-\gamma^0)\frac{1-|u|}{1+2u\gamma'}(\gamma'-\gamma^0)\frac{1}{1+2u\gamma^0}\Bigg)\Bigg|\leq \frac{C}{\eta}\norm{\gamma-\gamma'}_{\gS_2(\gH_\Lambda)}
\end{multline*}
using in particular
$$\frac{1-|u|}{1+2u\gamma}-\frac{1-|u|}{1+2u\gamma'}=2u\frac{1-|u|}{1+2u\gamma}(\gamma'-\gamma)\frac{1}{1+2u\gamma'}$$
and $0\leq (1+2u\gamma')^{-1}\leq \eta^{-1}$ as $\gamma'\in\cK$ and $-1+\eta\leq u\leq 1-\eta$. Similarly
$$\Bigg| \tr\int_{1-\eta}^1du\frac{1}{1+2u\gamma^0}(\gamma-\gamma^0)\frac{1-|u|}{1+2u\gamma}(\gamma-\gamma^0)\frac{1}{1+2u\gamma^0}\Bigg|\leq C\eta\norm{\gamma-\gamma^0}_{\gS_2(\gH_\Lambda)}^2.$$
The other terms are treated in the same way.
\end{proof}

Convexity of $\gamma \mapsto H(\gamma,\gamma_0)$ is a simple consequence of the integral representation (\ref{def_H_int}). In fact, the integrand is convex for any fixed $u\in [-1,1]$, since
\begin{align}\nonumber
\gamma & \mapsto (\gamma-\gamma^0) \frac1{1+2u\gamma}(\gamma-\gamma^0) \\ \nonumber &\qquad = \frac 1{(2u)^2} \left( 2u \gamma - 1 - 4u\gamma_0 + (1+2u\gamma_0)\frac 1{1+2u\gamma}(1+2u\gamma_0)\right)
\end{align}
is clearly convex.

Finally, we prove Formula \eqref{estim_below_H}. Consider the following function
$$f(x,y)=(\half+x)\left(\ln(\half+x)-\ln(\half+y)\right)+(\half-x)\left(\ln(\half-x)-\ln(\half-y)\right)$$
defined on $(-1/2,1/2)^2$. Minimizing over $x$ for fixed $y$, one
finds that $f(x,y)\geq (x-y)^2C(y)$ where
$C(y)=\ln\left(\frac{1/2+y}{1/2-y}\right)/(2y)$. If we write $y$ as
\begin{equation}
 y=\frac12\left(\frac{1}{1+e^h}-\frac{1}{1+e^{-h}} \right),
\label{form_y}
\end{equation}
we obtain
$C(y)=h\,\text{tanh}(h/2)^{-1}\geq\max(|h|,2)$.
Hence if $y$ takes the form \eqref{form_y}, we deduce
$$f(x,y)\geq \max\left\{(x-y)^2|h|,\, 2(x-y)^2\right\}.$$
Assume now that $X$ and $Y$ are self-adjoint operators acting on a Hilbert space $\gH$, with $-1/2\leq X,Y\leq 1/2$ and
$$Y=\frac12\left(\frac{1}{1+e^H}-\frac{1}{1+e^{-H}} \right)$$
for some $H$. By Klein's inequality \cite[p. 330]{Thirring}, one also has
\begin{equation}
 H(X,Y)=\tr f(X,Y)\geq \max\big\{\tr(X-Y)^2|H|\, ,\, 2\tr(X-Y)^2\big\}.
\label{Klein_inequality}
\end{equation}
This gives \eqref{estim_below_H}, taking $X=\gamma$ and $Y=\gamma^0$.
\end{proof}

\subsection{Existence of a minimizer and Debye screening}\label{sec:external_red}
Now we are able to define the \emph{reduced Bogoliubov-Dirac-Fock} energy at temperature $T=\beta^{-1}$. For this purpose, we introduce the Coulomb space
\begin{equation}
 \cC:=\{\rho \in \mathcal{S}'(\R^3) \ | \  D(\rho,\rho)<\ii\}
\label{def_cC}
\end{equation}
where
\begin{equation}
 D(f,g)=4\pi\int_{\R^3}|k|^{-2}\overline{\widehat{f}(k)}g(k)dk.
\label{def_Coulomb}
\end{equation}
We remark that the Fourier transform of $Q=\gamma-\gamma^0$ in an
$L^2$-function with support in $B(0,\Lambda)\times B(0,\Lambda)$.
Hence $Q(x,y)$ is a smooth kernel and $\rho_Q(x)=\tr_{\C^4}(Q(x,x))$
is a well defined function. Indeed, the map
$\gamma\in\cK\mapsto\rho_{\gamma-\gamma^0}\in L^2(\R^3)$ is continuous
for the topology of $\gS_2(\gH_\Lambda)$. It is easy to see that the
Fourier transform of $\rho_{\gamma-\gamma^0}$ is given by the formula
\begin{equation}
\widehat{\rho_{\gamma-\gamma^0}}(k)=\frac{1}{(2\pi)^{3/2}}\int_{\substack{|p+k/2|\leq\Lambda\\ |p-k/2|\leq\Lambda}}
\tr_{\C^4}\left[\widehat{(\gamma-\gamma^0)}(p+k/2,p-k/2)\right]dp.
\label{def_rho}
\end{equation}

We also define our variational set by
\begin{equation}
\cK_C:=\left\{\gamma\in\cK\ |\ \rho_{\gamma-\gamma^0}\in\cC\right\}.
\label{def_convex_set2}
\end{equation}
The \emph{reduced Bogoliubov-Dirac-Fock energy} reads
\begin{equation}\label{def_rBDF}
 \boxed{\cF^{\rm red}_T(\gamma) = T H(\gamma, \gamma^0)   - \alpha D(\nu,\rho_{\gamma - \gamma^0})+ \frac \alpha 2
 D(\rho_{\gamma - \gamma^0},\rho_{\gamma - \gamma^0})}
 \end{equation}
and it is well-defined on $\cK_C$ by Theorem \ref{prop_rel_entropy}. In \eqref{def_rBDF}, $\nu\in\cC$ is an external density creating an electrostatic potential $-\nu\ast1/|x|$. The number $\alpha>0$ is the \emph{fine structure constant}. The following is an easy consequence on Theorem \ref{prop_rel_entropy}:
\begin{thm}[{\bf Existence of a minimizer}]\label{thm_exists_rBDF} Assume $T>0$, $\alpha\geq0$ and $\nu\in\cC$. Then $\cF^{\rm red}_T$ satisfies
\begin{equation}
\forall\gamma\in\cK_C,\qquad \cF^{\rm red}_T(\gamma)\geq -\frac\alpha2 D(\nu,\nu)
\label{estim_below_red_BDF}
\end{equation}
hence it is bounded below on $\cK_C$. It has a \emph{unique minimizer} $\bar\gamma$ on $\cK_C$. The operator $\bar\gamma$ satisfies the self-consistent equation
\begin{equation}
\left\{\begin{array}{l}\displaystyle
\bar\gamma = \frac12\left(\frac1{1+e^{\beta D_{\bar\gamma}}}-\frac{1}{1+e^{-\beta D_{\bar\gamma}}} \right), \smallskip\\
D_{\bar\gamma} :=  D^0+\alpha(\rho_{\bar\gamma-\gamma^0}-\nu)\ast|\cdot|^{-1}.
\end{array}\right.
\label{SCF_red_BDF}
\end{equation}
\end{thm}

\begin{remark}
When $T=0$, a similar result was proved in \cite[Thm 3]{HLS2}, but
there might be no uniqueness in this case.
\end{remark}

\begin{remark}
If there is no external field, $\nu=0$, we recover that the optimal state is $\gamma-\gamma^0=0$, and its energy is zero, by \eqref{estim_below_red_BDF}.
\end{remark}

\begin{proof}[Proof of Theorem \ref{thm_exists_rBDF}]
  Eq.~\eqref{estim_below_red_BDF} is an obvious consequence of
  positivity of the relative entropy $H$ and positive definiteness of
  $D(\cdot,\cdot)$.  The existence of a minimizer is obtained by noticing that
  $\cF^{\rm red}_T$ is weakly lower semi-continuous for the topology
  of $\gS_2(\gH_\Lambda)$ and $\cC$, by Theorem
  \ref{prop_rel_entropy}. As $\cF^{\rm red}_T$ is convex\footnote{It
    can indeed be proved that $H(\cdot,\gamma^0)$ is strictly convex
    but we do not need that here.} and strictly convex with respect to
  $\rho_{\gamma-\gamma^0}$, we deduce that all the minimizers share
  the same density. Next we notice that
  $\pm1/2\notin\sigma(\bar\gamma)$ since the derivative of the
  relative entropy with respect to variations of an eigenvalue is
  infinite at these two points. Hence $\bar\gamma$ does not saturate
  the constraint and it is a solution of Eq.~\eqref{SCF_red_BDF}. This
  \emph{a fortiori} proves that $\bar\gamma$ is unique, since $D_{\bar
    \gamma}$ depends only on the density $\rho_{\bar\gamma-\gamma^0}$.
\end{proof}

Now we provide some interesting properties of any solution of Eq.~\eqref{SCF_red_BDF}, thus in particular of our minimizer $\bar\gamma$.
\begin{thm}[{\bf Debye Screening}]\label{thm_Debye_rBDF} Assume $T>0$,
  $\alpha>0$ and $\nu\in\cC\cap L^1(\R^3)$. Any $\gamma\in\cK$ that
  solves Eq.~\eqref{SCF_red_BDF} is a trace-class perturbation of
  $\gamma^0$, i.e., $\gamma-\gamma^0\in \gS_1(\gH_\Lambda)$. Its
  charge density $\rho_{\gamma-\gamma^0}$ is an $L^1(\R^3)$ function
  which satisfies
\begin{equation}
\int_{\R^3}\rho_{\gamma-\gamma^0}=\int_{\R^3}\nu\qquad\text{and}\qquad \big(\rho_{\gamma-\gamma^0}-\nu\big)\ast\frac{1}{|x|}\in L^1(\R^3).
\label{Debye}
\end{equation}
\end{thm}

This result implies that the particles arrange themselves such that
the total effective potential $\big(\rho_{\gamma-\gamma^0} -
\nu\big)\ast 1/|x|$ has a decay much faster than $1/|x|$. This implies that
the nuclear charge of the external sources is completely screened.

The proof of Theorem \ref{thm_Debye_rBDF} is lengthy and is given
later in Section \ref{sec:proof_Debye}.

\section{The Bogoliubov-Dirac-Fock free energy }
\subsection{Definition of the free vacuum}
When the exchange term is not neglected, the free vacuum is no longer
described by the operator $\gamma^0$ introduced in the previous
section. Instead it is another translation-invariant operator
$\tilde\gamma^0$ that solves a self-consistent equation. Following
ideas from \cite{HLSo}, we define in this section $\tilde\gamma^0$ as
the (unique) minimizer of the free energy per unit volume.  We
consider translation-invariant operators $\gamma=\gamma(p)$ acting on
$\gH_\Lambda$ and such that $-1/2\leq\gamma\leq 1/2$ which is
obviously equivalent to $-1/2\leq\gamma(p)\leq 1/2$, for a.e. $p\in
B(0,\Lambda)$, in the sense of $\C^4\times\C^4$ hermitian matrices.
The free energy per unit volume of such a translation-invariant
operator $\gamma$ at temperature $T$ is given by \cite{HLSo}
\begin{multline}
\cT_T(\gamma)=\frac{1}{(2\pi)^{3}}\bigg(\int_{B(0,\Lambda)}\tr_{\C^4}[
D^0(p)\gamma(p)]dp\\-\frac{\alpha}{(2\pi)^2}\iint_{B(0,\Lambda)^2}\frac{\tr_{\C^4}
[\gamma(p)\gamma(q)]}{|p-q|^2}dp\,dq
-TS(\gamma)\bigg)
\label{def_fn_F}
\end{multline}
where the entropy is defined as
$$S(\gamma)=-\int_{B(0,\Lambda)}\!\!\tr_{\C^4}\big[
\left(\half+\gamma(p)\right)\ln\left(\half+\gamma(p)\right)+\left(\half-\
\gamma(p)\right)\ln\left(\half-\gamma(p)\right)\big]dp.$$ The free
energy is defined on the convex set of matrix-valued functions, such
that, for all $p\in B(0,\Lambda)$,  $\gamma(p)$ is a hermitian $4\times 4$ matrix, i.e.
\begin{multline}\label{defcA}
\cA:=\big\{\gamma :B(0,\Lambda) \to M^4 \, | \,
\gamma(p)^*=\gamma(p),\\ -1/2\leq\gamma(p)\leq1/2 \ \text{for all $p\in B(0,\Lambda)$}\big\}.
\end{multline}

\begin{thm}[{\bf The free vacuum at temperature $T$}]\label{thm_free_case} For all $T>0$
and all $0\leq\alpha<4/\pi$, the free energy per unit volume $\cT_T$ in (\ref{def_fn_F}) has a
unique minimizer $\tilde\gamma^0$ on $\cA$. It is a solution of the self-consistent
equation
\begin{equation}
\left\{\begin{array}{l}
\displaystyle \tilde\gamma^0=\frac12\left( \frac{1}{1+e^{ \beta D_{\tilde\gamma^0}}}-\frac{1}{1+e^{-\beta D_{\tilde\gamma^0}}}\right) \\
D_{\tilde\gamma^0}=D^0-\alpha\frac{\tilde\gamma^0(x,y)}{|x-y|}.
\end{array}\right.
\label{SCF_free}
\end{equation}
Furthermore, $\tilde\gamma^0$ has the form
\begin{equation}
 \tilde\gamma^0(p)=f_1(|p|)\alp\cdot p+f_0(|p|)\beta
\label{form_gamma_0}
\end{equation}
with $f_0,f_1\leq0$ a.e. on $B(0,\Lambda)$ and $D_{\tilde\gamma^0}$ satisfies
\begin{equation}
 |D_{\tilde\gamma^0}|\geq|D^0|.
\label{bound_below_free}
\end{equation}
\end{thm}

Here and in the following, we shall identify operators with their integral kernels for simplicity of the notation. That is, the last term in the second line of
\eqref{SCF_free} denotes the operator  with integral kernel
given by $ \frac{(\gamma-\tilde\gamma^0)(x,y)}{|x-y|}$, where
$(\gamma-\tilde\gamma^0)(x,y)$ is the integral kernel of the
translation-invariant operator $\gamma -\tilde\gamma^0$ (it is a function of $x-y$).

\begin{remark}
  The assumption $\alpha < 4/\pi$ guarantees that the functional
  \eqref{def_fn_F} is bounded from below, independently of the UV
  cutoff $\Lambda$, which is arbitrary in this paper. This is a
  consequence of Kato's inequality. For $\alpha>4/\pi$ this is not the
  case \cite{CIL,noebsi}.

  For comparison, we note that in the \emph{non-interacting} case $\alpha =
  0$, the functions $f_1(|p|)$ and $f_0(|p|)$ appearing in Theorem
  \ref{thm_free_case} are given by
$$f_1(|p|)= f_0(|p|) = \frac 1{2E(p)} \left( \frac{1}{1+e^{ \beta E(p)}}-\frac{1}{1+e^{-\beta E(p)}} \right).$$
\end{remark}

A similar result was proved in the zero temperature case in
\cite{HLSo}. As in \cite{HLSo}, it is possible to justify the
introduction of $\cT_T$ by a thermodynamic limit procedure. The proof
of Theorem \ref{thm_free_case} is given in Section
\ref{sec:proof_free_case}.

Like for the reduced case, we have that
$$\sigma(\tilde{\gamma}^0)\subset\left[-\frac12+\epsilon,-\epsilon\right]\cup\left[\epsilon,\frac12-\epsilon\right]$$
for some $\epsilon>0$. This can be seen from \eqref{bound_below_free}
and the fact that $D_{\tilde{\gamma}^0}$ is a bounded operator on
$\gH_\Lambda$ due to the presence of the ultraviolet cut-off. Notice
also that we have formally $\rho_{\tilde{\gamma}^0}\equiv0$ by
\eqref{form_gamma_0}, as in \eqref{rho_vanishes}.

\subsection{The external field case}\label{sec:external}
As in Section \ref{sec:external_red}, one can consider the Bogoliubov-Dirac-Fock energy with an external field. It is formally obtained by subtracting the infinite free energy of the free vacuum at temperature $T>0$ from the free energy of our state $\gamma$. This procedure can be justified like in \cite{HLSo} by a thermodynamic limit procedure.
Using the same notation as in Section~\ref{sec:external_red}, the \emph{Bogoliubov-Dirac-Fock free energy} reads
\begin{multline}\label{BDF-functional}
 \cF_T(\gamma) = T H(\gamma, \tilde\gamma^0)   - \alpha D(\nu, \rho_{\gamma-\tilde\gamma^0}) + \frac \alpha 2
 D(\rho_{\gamma-\tilde\gamma^0},\rho_{\gamma-\tilde\gamma^0})\\
 - \frac \alpha 2  \iint \frac{\tr_{\C^4}|(\gamma - \tilde\gamma^0)(x,y)|^2}{|x-y|}dx dy,
 \end{multline}
where $H$ is the relative entropy defined like in Section \ref{sec:def_H}.
Like for the reduced case, we see that the functional $\cF_T$ is well-defined on the following convex set
\begin{equation}
\tilde\cK_C:=\left\{\gamma\in\cB(\gH_\Lambda)\ |\ \gamma^*=\gamma,\ -\frac{1}{2}\leq\gamma\leq \frac{1}{2},\ \gamma-\tilde\gamma^0\in\gS_2(\gH_\Lambda),\ \rho_{\gamma-\tilde\gamma^0}\in\cC\right\}.
\label{def_convex_set2_BDF}
\end{equation}
Note that although the function $\gamma\mapsto H(\gamma, \tilde\gamma^0)$ is
convex, $\cF_T$ is not a convex functional because of the
presence of the exchange term. This is of course a great obstacle in
proving the existence of a minimizer, and we have to leave this as an open problem. Following the method of Theorem \ref{prop_rel_entropy}, we shall show that
\begin{equation}
\forall\gamma\in\tilde\cK_C,\qquad TH(\gamma,\tilde\gamma^0)\geq \tr\left(|D_{\tilde\gamma^0}|(\gamma-\tilde\gamma^0)^2\right).
\label{estim_below_H_BDF}
\end{equation}
With the aid of this inequality we can prove the
\begin{thm}[{\bf Minimizer in External Field}]\label{thm_prop_full}Assume that $0\leq\alpha<4/\pi$ and that $T>0$. We have
 \begin{equation}
\forall\gamma\in\tilde\cK_C,\qquad \cF_T(\gamma)\geq -\frac\alpha2 D(\nu,\nu)
\label{estim_below_BDF}
\end{equation}
and hence $\cF_T$ is bounded below on $\tilde\cK_C$.

Assume that $\gamma\in\tilde\cK_C$ is a minimizer of $\cF_T$. Then it satisfies the self-consistent equation
\begin{equation}
\left\{\begin{array}{l}\displaystyle
\gamma = \frac12\left(\frac1{1+e^{\beta D_{\gamma}}}-\frac{1}{1+e^{-\beta D_{\gamma}}} \right), \smallskip\\
D_{\gamma} :=  D_{\tilde\gamma^0}+\alpha(\rho_{\gamma-\gamma^0}-\nu)\ast|\cdot|^{-1}-\alpha\frac{(\gamma-\tilde\gamma^0)(x,y)}{|x-y|}
\end{array}\right.
\label{SCF_BDF}
\end{equation}
with $D_{\tilde\gamma^0}$ defined in (\ref{SCF_free}).
It is unique when
\begin{equation}
0\leq\alpha\frac\pi4\left\{1-\alpha\left(\frac\pi2 \sqrt{\frac{\alpha/2}{1-\alpha\pi/4}}+\pi^{1/6}2^{11/6}\right)D(\nu,\nu)^{1/2}\right\}^{-1}\leq1.
\label{condition_uniqueness}
\end{equation}
\end{thm}
 The proof of Theorem
\ref{thm_prop_full} is provided in Section
\ref{sec:proof_thm_prop_full}.

\section{Proofs}
\subsection{Proof of Theorem \ref{thm_Debye_rBDF}}\label{sec:proof_Debye}
Let $\gamma$ be a solution of
\begin{equation}
\left\{\begin{array}{l}\displaystyle
\gamma = \frac12\left(\frac1{1+e^{\beta D_{\gamma}}}-\frac{1}{1+e^{-\beta D_{\gamma}}} \right), \smallskip\\
D_{\gamma} :=  D^0+\alpha(\rho_{\gamma-\gamma^0}-\nu)\ast|\cdot|^{-1}.
\end{array}\right.
\label{SCF_proof}
\end{equation}
For the sake of simplicity, we define $\rho:=\rho_{\gamma-\gamma^0}-\nu$ and $V=\alpha(\rho_{\gamma-\gamma^0}-\nu)\ast|\cdot|^{-1}$. Note that  $\nabla V\in L^2(\R^3)$ as $\rho\in\cC$, hence $V\in L^6(\R^3)$.
Following \cite[p. 4495]{HLS2}, we may use the Kato-Seiler-Simon inequality (see \cite{SeSi} and \cite[Thm.~4.1]{Simon})
\begin{equation}
 \forall p\geq2,\qquad \norm{f(-i\nabla)g(x)}_{\gS_p(L^2(\R^3))}\leq (2\pi)^{-3/p}
\norm{g}_{L^p(\R^3)}\norm{f}_{L^p(\R^3)}
\label{KSS}
\end{equation}
to obtain
$$\norm{V\frac{1}{|D^0|}}_{\gS_\ii(\gH_\Lambda)}\leq \norm{V\frac{1}{|D^0|}}_{\gS_6(\gH_\Lambda)}\leq C'\norm{V}_{L^6(\R^3)}\leq C\norm{\rho}_\cC.$$
This shows that
$|D_\gamma|\leq (1+\alpha C\norm{\rho}_\cC)|D^0|.$
Thanks to the cut-off in Fourier space, we deduce that $D_\gamma$ is a bounded operator or $\gH_\Lambda$.
Recall Duhamel's formula
\begin{equation}
e^{\beta D_\gamma} = e^{\beta D^0}+ \beta\int_0^1 e^{t \beta D_\gamma} V
e^{(1-t)\beta D^0} dt.
\label{Duhamel}
\end{equation}
Denoting $K :=  \beta\int_0^1 e^{t\beta  D_\gamma}V e^{(1-t)\beta D^0} dt$ and using \eqref{Duhamel}, we have
\begin{align*}
K&=K_0+K'\\
&:=\beta\int_0^1dt\; e^{t\beta  D^0}V e^{(1-t)\beta D^0} +\beta^2\int_0^1dt\int_0^tds\; e^{s\beta  D_\gamma}V e^{(t-s)\beta D^0}Ve^{(1-t)\beta D^0}.
\end{align*}
We obtain for the self-consistent solution
\begin{align}
\gamma - \gamma^0&= \frac 1{1 + e^{\beta D_\gamma}} - \frac 1{1+ e^{\beta D^0}}\nonumber \\
&= - \frac 1{1+ e^{\beta D^0}} K_0 \frac 1{1+ e^{\beta D^0}} - \frac 1{1+ e^{\beta D^0}} K' \frac 1{1+ e^{\beta D^0}}\nonumber\\
& \qquad\qquad\qquad\qquad+ \frac 1{1+ e^{\beta D^0} }K\frac 1{1+ e^{\beta D_\gamma}} K \frac 1{1+ e^{\beta D^0} }\label{SCF_proof_decomp}
\end{align}
which we write as $\gamma - \gamma^0= A+B$ where
$$A=- \frac 1{1+ e^{\beta D^0}} K_0 \frac 1{1+ e^{\beta D^0}}=-\beta\int_0^1 \frac{e^{t\beta  D^0}}{1+ e^{\beta D^0}}V \frac{e^{(1-t)\beta D^0}}{1+ e^{\beta D^0}} dt.$$
As $V\in L^6(\R^3)$ and $D_\gamma$ is bounded, using the cut-off in Fourier space and the Kato-Seiler-Simon inequality \eqref{KSS}, we have  $K\in\gS_6(\gH_\Lambda)$. Hence we obtain that $K'\in\gS_3(\gH_\Lambda)$ and $B\in\gS_3(\gH_\Lambda)$.

The next step is to compute the density of $A$. The kernel of $A$ is given by
$$\widehat{A}(p,q)=-\beta(2\pi)^{-3/2}\int_0^1 \frac{e^{t\beta  D^0(p)}}{1+ e^{\beta D^0(p)}}\widehat{V}(p-q) \frac{e^{(1-t)\beta D^0(q)}}{1+ e^{\beta D^0(q)}} dt.$$
Using \eqref{def_rho}, we obtain
$$\widehat{\rho_{A}}(k)=-\frac{\alpha C(|k|)}{|k|^2}\widehat{\rho}(k)$$
where
\begin{equation}\label{defC}
C(|k|):=\frac{\beta}{2\pi^2}\int_{\substack{|p+k/2|\leq\Lambda\\ |p-k/2|\leq\Lambda}}dp\int_0^1dt\; \tr_{\C^4}\left[\frac{e^{t\beta  D^0(p+k/2)}}{1+ e^{\beta D^0(p+k/2)}} \frac{e^{(1-t)\beta D^0(p-k/2)}}{1+ e^{\beta D^0(p-k/2)}}\right].
\end{equation}
Inserting this into the self-consistent equation \eqref{SCF_proof_decomp} gives
\begin{equation}
\widehat{\rho}(k)=-\widehat{\nu}(k)-\frac{\alpha C(|k|)}{|k|^2}\widehat{\rho}(k)+\widehat{\rho_B}(k)
\end{equation}
or, equivalently,
\begin{equation}
\widehat{\rho}(k)=\widehat{b_1}(k)
(-\widehat{\nu}(k)+\widehat{\rho_B}(k)), \label{SCF_eq_for_rho}
\end{equation}
and
\begin{equation}
\widehat{V}(k)=4\pi \widehat{b_2}(k)
(-\widehat{\nu}(k)+\widehat{\rho_B}(k)), \label{SCF_eq_for_V}
\end{equation}
where
\begin{equation}\label{b_1b_2} b_1:=\cF^{-1}\left(\frac{|k|^2}{|k|^2+\alpha
C(|k|)}\right)\quad\text{and}\quad
b_2:=\cF^{-1}\left(\frac{1}{|k|^2+\alpha C(|k|)}\right),
\end{equation}
with $\cF^{-1}$ denoting the inverse Fourier transform.
 Our main
tool will be the following
\begin{prop}[{\bf Properties of $b_1,b_2$}]\label{prop_C}
The two functions $b_1(x)$ and $b_2(x)$, defined in \eqref{b_1b_2} and \eqref{defC},
belong to $L^1(\R^3)$.
\end{prop}
We postpone the proof of Proposition \ref{prop_C} to Section
\ref{sec:proof_prop_C} and first complete the proof of Theorem
\ref{thm_Debye_rBDF}. First we claim that $\rho_B\in L^3(\R^3)$. To
see this, we take a function $\xi\in L^{3/2}(\R^3)\cap C^\ii_0(\R^3)$
and compute
\begin{align*}
  |\tr(B\xi)|&=|\tr(B\1_{B(0,\Lambda)}(p)\xi\1_{B(0,\Lambda)}(p))|\\
  &\leq
  \norm{B}_{\gS_3(\gH_\Lambda)}\norm{\1_{B(0,\Lambda)}(p)\xi\1_{B(0,\Lambda)}(p)}_{\gS_{3/2}(\gH_\Lambda)}.
\end{align*}
Writing $\xi=|\xi|^{1/2}\text{sgn}(\xi)|\xi|^{1/2}$ and using the
Kato-Seiler-Simon inequality \eqref{KSS} twice in
$\gS_{3}(\gH_\Lambda)$, we obtain $|\tr(B\xi)|\leq
C\norm{B}_{\gS_3(\gH_\Lambda)}\norm{\xi}_{L^{3/2}(\R^3)}$ where $C$
depends on the cut-off $\Lambda$. This proves by duality that
$\rho_B\in L^3(\R^3)$.

Next we use a boot-strap argument. As $\nu\in L^1(\R^3)$ and
$\rho_B\in L^3(\R^3)$, we get from \eqref{SCF_eq_for_V} and
Proposition \ref{prop_C} that $V\in L^3(\R^3)$. Inserting in the
definition of $K'$ and using \eqref{KSS} once more, we obtain that
$K'\in\gS_{3/2}(\gH_\Lambda)$, hence $B\in \gS_2(\R^3)$ and $\rho_B\in
L^2(\R^3)$. Using again \eqref{SCF_eq_for_V} and Proposition
\ref{prop_C}, we get that $V\in L^2(\R^3)$, hence
$B\in\gS_1(\gH_\Lambda)$ and $\rho_B\in L^1(\R^3)$. This finishes the
proof of Theorem \ref{thm_Debye_rBDF}, by \eqref{SCF_eq_for_rho},
\eqref{SCF_eq_for_V} and Proposition \ref{prop_C}.\qed

\subsection{Proof of Proposition \ref{prop_C}}\label{sec:proof_prop_C}
The proof proceeds along the same lines as in the Appendix of
\cite{GLS}. In the following we shall denote by $P_0^+$ and $P_0^-$
the projection onto the positive and negative spectral subspace of
$D^0$, respectively. As multiplication operators in momentum space,
$$P_0^+(p) = \frac 12 \left(1 + \frac{\alp\cdot p + \beta}{E(p)}\right) \quad ,\
P_0^-(p) = \frac 12 \left(1 - \frac{\alp\cdot p +
\beta}{E(p)}\right).$$ The function $C$ in (\ref{defC}) can be written as
\begin{multline}
C(|k|)=\frac{\beta}{\pi^2}\int_{\substack{|p+k/2|\leq\Lambda\\ |p-k/2|\leq\Lambda}}dp\times\\
\times\int_0^1dt\Bigg(
 \tr_{\C^4}\left[\frac{e^{t\beta  E(p+k/2)}}{1+ e^{\beta E(p+k/2)}} \frac{e^{(1-t)\beta E(p-k/2)}}{1+ e^{\beta E(p-k/2)}}P^+_0(p+k/2)P^+_0(p-k/2)\right]\\
+\tr_{\C^4}\left[\frac{e^{t\beta  E(p+k/2)}}{1+ e^{\beta E(p+k/2)}} \frac{e^{-(1-t)\beta E(p-k/2)}}{1+ e^{-\beta E(p-k/2)}}P^+_0(p+k/2)P^-_0(p-k/2)\right]\Bigg).
\end{multline}
Hence
\begin{align*}
&C(|k|)=\\
&\qquad\frac{1}{\pi^2}\int_{\substack{|p+k/2|\leq\Lambda\\ |p-k/2|\leq\Lambda}}
\frac{1}{1+ e^{\beta E(p+k/2)}}\frac{e^{\beta E(p+k/2)}-e^{\beta E(p-k/2)}}{E(p+k/2)- E(p-k/2)} \frac{1}{1+ e^{\beta E(p-k/2)}}\times\\
&\qquad\qquad\qquad\qquad\times\left[1+\frac{(p+k/2)\cdot(p-k/2)+1}{E(p+k/2)E(p-k/2)}\right]dp\\
& + \frac{1}{\pi^2}\int_{\substack{|p+k/2|\leq\Lambda\\ |p-k/2|\leq\Lambda}}
\frac{1}{1+ e^{\beta E(p+k/2)}}\frac{e^{\beta E(p+k/2)}-e^{-\beta E(p-k/2)}}{E(p+k/2)+E(p-k/2)} \frac{1}{1+ e^{-\beta E(p-k/2)}}\times\\
&\qquad\qquad\qquad\qquad\times\left[1-\frac{(p+k/2)\cdot(p-k/2)+1}{E(p+k/2)E(p-k/2)}\right]dp.
\end{align*}
For the sake of clarity, we denote by $C_1(|k|)$ (resp. $C_2(|k|)$) the first (resp. second) integral of the previous formula. By the monotonicity of the exponential function, it is easily seen that $C_1(|k|)\geq0$ and $C_2(|k|)\geq0$.

The next step is to simplify the above integral formula. We follow a method of Pauli and Rose \cite{PauliRose} which was recently used in the appendix of \cite{GLS}. After two changes of variables, we end up with
\begin{align}
&C_1(|k|)=\frac{8}{\pi|k|}\int_0^{Z_\Lambda(|k|)}\!\!dz\int_0^{\frac{|k|}2z}\!\!dv\frac{e^{\beta w(k,z)}}{1+e^{\beta (w(k,z)+v)}}\frac{\sinh(\beta v)}{v}\times\nonumber\\
& \qquad\qquad\qquad\qquad\qquad\qquad\qquad\qquad\qquad\times\frac{1}{1+e^{\beta (w(k,z)-v)}}\frac{z\;w(k,z)^{-1}}{(1-z^2)^3}\nonumber\\
&\qquad+\frac{8}{\pi|k|}\int_0^{\frac{|k|}2Z_\Lambda(|k|)}\!\!dz\int_0^{z}\!\!dv\frac{e^{\beta (E(\Lambda)-z)}}{1+e^{\beta (E(\Lambda)-z+v)}}\frac{\sinh(\beta v)}{v}\frac{1}{1+e^{\beta (E(\Lambda)-z-v)}}\times\nonumber\\
& \qquad\qquad\qquad\qquad\qquad\qquad\qquad\qquad\times\left((E(\Lambda)-z)^2-\frac{|k|^2}{4}\right),\label{formula_C1}
\end{align}
\begin{align}
&C_2(|k|)=\frac{8}{\pi|k|}\int_0^{Z_\Lambda(|k|)}\!\!dz\int_0^{\frac{|k|}2z}\!\!dv\frac{e^{\beta v}}{1+e^{\beta (w(k,z)+v)}}\frac{\sinh(\beta w)}{1+\frac{|k|^2}4(1-z^2)}\times\nonumber\\
& \qquad\qquad\qquad\qquad\qquad\qquad\qquad\qquad\times\frac{1}{1+e^{\beta (v-w(k,z))}}\left(\frac{|k|^2}{4}-v^2\right)\frac{z}{1-z^2}\nonumber\\
&\qquad+\frac{8}{\pi|k|}\int_0^{\frac{|k|}2Z_\Lambda(|k|)}\!\!dz\int_0^{z}\!\!dv\frac{e^{\beta v}}{1+e^{\beta (E(\Lambda)-z+v)}}\frac{\sinh(\beta (E(\Lambda)-z))}{E(\Lambda)-z}\times\nonumber\\
& \qquad\qquad\qquad\qquad\qquad\qquad\qquad\times\frac{1}{1+e^{\beta (v+z-E(\Lambda))}}\left(\frac{|k|^2}{4}-v^2\right).\label{formula_C2}
\end{align}
In the above formulas we have used the notation (as in \cite{GLS})
$$Z_\Lambda(r)=\frac{\sqrt{1+\Lambda^2}-\sqrt{1+(\Lambda-r)^2}}{r}.$$
Note that $Z_\Lambda$ is a decreasing $C^\ii$ function on $[0,2\Lambda]$ satisfying $Z_\Lambda(0)=\Lambda/E(\Lambda)$, $Z_\Lambda(2\Lambda)=0$. We have also used the shorthand notation
$$w(k,z)=\sqrt{\frac{1+|k|^2(1-z^2)/4}{1-z^2}}.$$
All integrands of the above formulas are real analytic functions of $r=|k|$ on a neighborhood of $[0,2\Lambda]$. Also all the integrals vanish at $k=0$. We deduce that $C_1$ and $C_2$ are smooth functions on $[0,2\Lambda]$. Using $Z_\Lambda(2\Lambda)=0$, one also sees that $C_1(2\Lambda)=C_1'(2\Lambda)=C_2(2\Lambda)=C_2'(2\Lambda)=0$. A Taylor expansion of the first integral of $C_1$ yields
$$C_1(0)=\frac4\pi \beta\int_1^{E(\Lambda)}\frac{t^2dt}{(1+e^{-\beta t})(1+e^{\beta t})}>0.$$
The end of the proof of Proposition \ref{prop_C} is then the same as in \cite[Prop. 17]{GLS}. First we notice that as $C(r)$ is bounded and has a compact support, $b_1$ and $b_2$ are in $L^\ii(\R^3)$. We now prove that they decay at least like $|x|^{-4}$  at infinity meaning that they also belong to $L^1(\R^3)$. To this end we write for $b=b_1$ or $=b_2$ the inverse Fourier transform in radial coordinates:
\begin{equation}
\forall x \in \R^3 \setminus \{ 0 \},\quad  b(x) = \frac{1}{\sqrt{2 \pi} |x|} \int_0^{2 \Lambda} (r\widehat{b}(r)) \sin(r |x|)\, dr.
\label{formula_Fourier_inverse_radial}
\end{equation}
 Integrating by parts and using
$\widehat{b}(2\Lambda)=\widehat{b}'(2\Lambda)=0$ yields
\begin{multline}
\forall x \in \R^3 \setminus \{ 0 \},\quad  b(x) = \frac{1}{\sqrt{2 \pi} |x|^4}\Bigg(2\Lambda \widehat{b}''(2\Lambda)\cos(2\Lambda|x|)-2\widehat{b}'(0)\\
- \int_0^{2 \Lambda} (r\widehat{b})^{(3)}(r) \cos(r |x|)\, dr\Bigg).
\label{formula_Fourier_inverse_radial2}
\end{multline}
This completes the proof of Proposition \ref{prop_C}.
\qed

\subsection{Proof of Theorem \ref{thm_free_case}}\label{sec:proof_free_case}
The proof is inspired by ideas from \cite{HLSo}.
We denote $I:=\inf_{\gamma\in\cA}\cT_T(\gamma)$.
We start by introducing the following auxiliary minimization problem
\begin{equation}
 J=\inf_{\gamma\in\cB}\cT_T(\gamma)
\label{auxi_pb}
\end{equation}
where $\cB\subset\cA$ is given by
\begin{equation}\label{defcB}
\cB:=\left\{\gamma\in\cA,\ \gamma(p)=f_1(|p|)\alp\cdot p+f_0(|p|)\beta,\
f_0,f_1\leq0\right\}.
\end{equation}

\begin{lemma}\label{lemma_exists_min_aux_pb}
There exists a minimizer $\tilde\gamma^0\in \cB$ for \eqref{auxi_pb}.
\end{lemma}
\begin{proof}[Proof of Lemma \ref{lemma_exists_min_aux_pb}] The functional
$\cT_T$ is weakly lower semi-continuous for the weak-$\ast$ topology of
$L^\ii(B(0,\Lambda))$. This is because $-S$ is convex and the exchange term is
continuous for the weak topology of $L^2(B(0,\Lambda))$ as shown in \cite{HLSo}.
Also $\cB$ is a bounded closed convex subset of $L^\ii(B(0,\Lambda))$. Hence
there exists a minimum.
\end{proof}
\begin{lemma}\label{lemma_estim_gamma}
Let  $\tilde\gamma^0\in\cB$ be a minimizer of \eqref{auxi_pb}. Then there exists an $\epsilon>0$ such that
$|\tilde\gamma^0|\leq 1/2-\epsilon$.
\end{lemma}
\begin{proof}
For $x\in [1/2,1/2]$,
$$s(x):= \left(\half+x\right)\ln\left(\half+x\right)+\left(\half-x\right)\ln\left(\half-x\right)$$
is an even function of $x$. Because of the special form of
$\tilde\gamma^0$, we have $\tilde\gamma^0(p)^2=\|\tilde\gamma^0(p)\|^2
I_{\C^4}$ for all $p\in B(0,\Lambda)$, where $\|\,\cdot\,\|$ denotes the matrix norm. Hence
$$\forall\gamma\in\cB,\quad S(\gamma)=-4\int_{B(0,\Lambda)}s(\|\gamma(p)\|)dp.$$
The derivative of $s$ is infinite at $x=1/2$ and the derivative of the terms of the first line of \eqref{def_fn_F} stays bounded. It is therefore clear that
$\{p\in B(0,\Lambda)\ |\ 1/2-\epsilon\leq\|\tilde\gamma^0(p)\|\leq1/2\}$
has zero measure for $\epsilon$ small enough.
\end{proof}

Let us now write the first order condition satisfied by $\tilde\gamma^0$. Since $\|\tilde\gamma^0(p)\|\leq 1/2-\epsilon$ for some $\epsilon$ small enough, we can consider a perturbation of the form
$$\gamma(p)=\tilde\gamma^0(p)+t\left(g_1(|p|)\alp\cdot p+g_0(|p|)\beta\right)$$
with $g_0,g_1\leq0$ and $t>0$ small enough. We obtain
\begin{equation}
\int_{B(0,\Lambda)}\tr_{\C^4}\left[\left(D_{\tilde\gamma^0}(p)+T\ln\frac{1/2+\tilde\gamma^0(p)}{1/2-\tilde\gamma^0(p)}\right)\big(g_1(|p|)\alp\cdot p+g_0(|p|)\beta\big) \right]dp\geq0
\label{1st_order_convex_cond}
\end{equation}
for all $g_1,g_0\leq0$.

We notice that the function $x\mapsto \ln\left(\frac{1/2+x}{1/2-x}\right)$ is odd, hence
\begin{equation}
\forall\gamma\in\cB,\quad
\ln\left(\frac{1/2+\gamma(p)}{1/2-\gamma(p)}\right)=
\text{sgn}(\gamma)\ln\left(\frac{1/2+\|\gamma(p)\|}{1/2-\|\gamma(p)\|}\right),
\label{equality_log}
\end{equation}
with $\text{sgn}(\gamma) = \gamma/|\gamma|$. We obtain that
$$\ln\left(\frac{1/2+\tilde\gamma^0(p)}{1/2-\tilde\gamma^0(p)}\right)=\tilde\gamma^0(p) F(\|\tilde\gamma^0(p)\|)$$
where
$F(x)=\ln\left(\frac{1/2+x}{1/2-x}\right)/x$.
On the other hand, we can write
$$D_{\tilde\gamma^0}=d_1(|p|)\alp\cdot p+d_0(|p|)\beta$$
where $d_1$ and $d_0$ are given by \cite[Eq. (72)-(73)]{HLSo}. Using
$f_1,f_0\leq0$, we immediately see that
\begin{equation}
 d_1(|p|)\geq 1\quad \text{ and }\quad d_0(|p|)\geq 1,
\label{prop_gs}
\end{equation}
which in particular proves that
\begin{equation}
\|D_{{\tilde\gamma^0}}(p)\|\geq \|D^0(p)\|\geq |p|.
\label{cond_D_aux_pb}
\end{equation}

All this gives
\begin{multline}
D_{\tilde\gamma^0}+T\ln\frac{1/2+\tilde\gamma^0}{1/2-\tilde\gamma^0}=\big(d_1(|p|)+Tf_1(|p|)F(\|\tilde\gamma^0(p)\|)\big)\alp\cdot p\\ + \big(d_0(|p|)+Tf_0(|p|)F(\|\tilde\gamma^0(p)\|)\big)\beta.
\end{multline}
Inserting this in \eqref{1st_order_convex_cond}, we obtain the first order conditions
\begin{equation}
 \left\{\begin{array}{ll}
d_1(|p|)+Tf_1(|p|)F(\|\tilde\gamma^0(p)\|)\leq0\\
d_0(|p|)+Tf_0(|p|)F(\|\tilde\gamma^0(p)\|)\leq0.
\end{array}\right.
\end{equation}
In particular, because of \eqref{prop_gs} we infer that
\begin{equation}
 \left\{\begin{array}{ll}
f_1(|p|)F(\|\tilde\gamma^0(p)\|)\leq-1/T\\
f_0(|p|)F(\|\tilde\gamma^0(p)\|)\leq-1/T.
\end{array}\right.
\label{inequality_fs}
\end{equation}
As $F(\|\tilde\gamma^0(p)\|)\geq0$ and $f_0\geq -\|\tilde\gamma^0(p)\|$, we obtain from \eqref{inequality_fs}
the inequality $|\gamma(p)|F(\|\gamma(p)\|)\geq1/T$. Hence
\begin{equation}
\|\gamma(p)\|\geq \frac{e^{1/T}-1}{2(1+e^{1/T})}.
\end{equation}
This inequality means that $f_0$ and $f_1$ cannot vanish simultaneously. But we can indeed prove that each of them cannot vanish, as expressed in the

\begin{lemma}\label{lemma_fs_dont_vanish}
Let  $\tilde\gamma^0(p)=f_1(|p|)\alp\cdot p+f_0(|p|)\beta$ be a minimizer of \eqref{auxi_pb}. Then there exists an
$\epsilon>0$ such that
$$f_0 \leq - \epsilon \quad and \quad f_1\leq-\epsilon.$$
\end{lemma}
\begin{proof}
By Lemma \ref{lemma_estim_gamma} we know that $\|\tilde\gamma^0(p)\|\leq 1/2-\epsilon$ for some $\epsilon>0$. By \eqref{inequality_fs} and the monotonicity of $F$ we obtain
$f_k\leq -\frac{1}{T\; F(1/2-\epsilon)}$
for $k=0,1$.
\end{proof}

\begin{lemma}
Let  $\tilde\gamma^0\in\cB$ be a minimizer of \eqref{auxi_pb}. Then it solves the self-consistent equation
\begin{equation}
 \tilde\gamma^0=\frac12\left(\frac{1}{1+e^{ \beta D_{\tilde\gamma^0}}}-\frac{1}{1+e^{-\beta D_{\tilde\gamma^0}}}\right).
\label{equation_min_aux_pb}
\end{equation}
\end{lemma}
\begin{proof}
As the constraints are not saturated by Lemmas \ref{lemma_estim_gamma} and \ref{lemma_fs_dont_vanish}, we obtain that the derivative vanishes, i.e.
\begin{equation}
 \left\{\begin{array}{ll}
d_1(|p|)+Tf_1(|p|)F(\|\tilde\gamma^0(p)\|)=0\\
d_0(|p|)+Tf_0(|p|)F(\|\tilde\gamma^0(p)\|)=0.
\end{array}\right.
\end{equation}
which means that
$$D_{\tilde\gamma^0}+T\ln\frac{1/2+\tilde\gamma^0}{1/2-\tilde\gamma^0}=0.$$
Hence $\tilde\gamma^0$ solves \eqref{equation_min_aux_pb}.
\end{proof}

Now we prove that the operator $\tilde\gamma^0$ defined in the
previous step is the unique minimizer of $\cT_T$ on the full space
$\cA$ defined in (\ref{defcA}), not merely on the subset $\cB$ in
(\ref{defcB}).  We have
$$\cT_T(\gamma)-\cT_T(\tilde\gamma^0)=TH(\gamma,\tilde\gamma^0)-\frac{\alpha}{(2\pi)^5}\iint_{B(0,\Lambda)^2}\frac{\tr_{\C^4}
[(\gamma-\tilde\gamma^0)(p)(\gamma-\tilde\gamma^0)(q)]}{|p-q|^2}dp\,dq$$
where $H$ is the relative entropy per unit volume
\begin{multline}
H(\gamma,\tilde\gamma^0)=(2\pi)^{-3}\int_{B(0,\Lambda)}\tr_{\C^4}\big[\left(\half+\gamma\right)\left(\ln\left(\half+\gamma\right)-\ln\left(\half+\tilde\gamma^0\right)\right)\\
+\left(\half-\gamma\right)\left(\ln\left(\half-\gamma\right)-\ln\left(\half-\tilde\gamma^0\right)\right)\big]\,dp. \label{hnew}
\end{multline}
We shall use the important
\begin{lemma}\label{lemma_bd_below_relative_entropy}
For $H$ in (\ref{hnew}) the inequality
\begin{equation}
 TH(\gamma,\tilde\gamma^0)\geq (2\pi)^{-3}\int_{B(0,\Lambda)}\tr_{\C^4}|D_{\tilde\gamma^0}(p)|(\gamma(p)-\tilde\gamma^0(p))^2\,dp
\label{bd_below_relative_entropy}
\end{equation}
holds for all $\gamma\in\cA$.
\end{lemma}

\begin{proof}
This is a simple application of \eqref{Klein_inequality}, taking $X=\gamma(p)$, $Y=\tilde\gamma^0(p)$ and integrating over the ball $B(0,\Lambda)$.
\end{proof}

Using Lemma \ref{lemma_bd_below_relative_entropy} and the formula
$$\frac{\alpha}{(2\pi)^2}\iint_{B(0,\Lambda)^2}\frac{\tr_{\C^4}
[\gamma(p)\gamma(q)]}{|p-q|^2}dp\,dq=\frac{\alpha}2\int_{\R^3}\frac{\tr_{\C^4}|\check{\gamma}(x)|^2}{|x|}dx$$
where $\check{\gamma}(x)$ is the Fourier inverse of the function $\gamma(p)$, we find
\begin{multline}
\cT_T(\gamma)-\cT_T(\tilde\gamma^0)\geq (2\pi)^{-3}\bigg(\int_{B(0,\Lambda)}\tr_{\C^4}|D_{\tilde\gamma^0}(p)|(\gamma(p)-\tilde\gamma^0(p))^2\,dp\\
- \frac{\alpha}2\int_{\R^3}\frac{\tr_{\C^4}|(\check{\gamma}-\check{\gamma}^0)(x)|^2}{|x|}dx\bigg)
\end{multline}
for all $\gamma\in\cA$.  We now use ideas of
\cite{BBHS,HLS1,HLSo}. Kato's inequality $|x|^{-1}\leq \pi/2|\nabla|$
gives
$$\frac{\alpha}2\int_{\R^3}\frac{\tr_{\C^4}|(\check{\gamma}-\check{\gamma}^0)(x)|^2}{|x|}dx\leq \frac{\alpha\pi}4\int_{B(0,\Lambda)}\tr_{\C^4}|p|(\gamma(p)-\tilde\gamma^0(p))^2\,dp.$$
By \eqref{cond_D_aux_pb} we deduce
$$\cT_T(\gamma)-\cT_T(\tilde\gamma^0)\geq (1-\pi\alpha/4)(2\pi)^{-3}\int_{B(0,\Lambda)}\tr_{\C^4}|D_{\tilde\gamma^0}(p)|(\gamma(p)-\tilde\gamma^0(p))^2\,dp.$$
Hence $\tilde\gamma^0$ is the unique minimizer of $\cT_T$ on $\cA$ when $0\leq\alpha<4/\pi$.
This completes the proof of Theorem \ref{thm_free_case}.\qed

\subsection{Proof of Theorem \ref{thm_prop_full}}\label{sec:proof_thm_prop_full}
The lower bound \eqref{estim_below_BDF} is obtained by following an argument of \cite{BBHS,HLS1}. By \eqref{estim_below_H_BDF} we have for all $\gamma\in\tilde\cK_C$
\begin{multline}
\cF_T(\gamma)\geq
\tr\left(|D_{\tilde\gamma^0}|(\gamma-\tilde\gamma^0)^2\right)-\frac
\alpha 2  \iint \frac{\tr_{\C^4}|(\gamma - \tilde\gamma^0)(x,y)|^2}{|x-y|}dx
dy\\+\frac\alpha2
D(\rho_{\gamma-\tilde\gamma^0}-\nu,\rho_{\gamma-\tilde\gamma^0}-\nu)-\frac\alpha2
D(\nu,\nu). \label{bound_below_Hardy}
\end{multline}
By \eqref{bound_below_free} together with Kato's inequality $|x|^{-1}\leq(\pi/2)|\nabla|\leq(\pi/2)|D^0| $,
$$\iint \frac{\tr_{\C^4}|(\gamma - \tilde\gamma^0)(x,y)|^2}{|x-y|}dx dy\leq \frac\pi2 \tr\left(|D_{\tilde\gamma^0}|(\gamma-\tilde\gamma^0)^2\right)$$
which yields \eqref{estim_below_BDF} when $0\leq\alpha<4/\pi$.

Assume now that $\gamma$ is a minimizer of $\cF_T$ on $\tilde\cK_C$. The proof that it satisfies the self-consistent equation \eqref{SCF_BDF} is the same as in the  case of the reduced BDF functional in Section~\ref{sec:external_red}. Note that because of \eqref{bound_below_Hardy} and $\inf_{\tilde\cK_C}\cF_T\leq0$, we obtain
$$\left(\frac2\pi - \frac\alpha2\right)\iint \frac{\tr_{\C^4}|(\gamma - \tilde\gamma^0)(x,y)|^2}{|x-y|}dx dy+\frac\alpha 2D(\rho_{\gamma-\tilde\gamma^0},\rho_{\gamma-\tilde\gamma^0})\leq \frac\alpha 2D(\nu,\nu).$$
It was proved in \cite[p. 4495]{HLS2} that this implies (under the condition \eqref{condition_uniqueness}) that
\begin{equation}
|D_\gamma|\geq d^{-1}|D^0|
\label{estim_below_scf_operator}
\end{equation}
with
$$d=\left\{1-\alpha\left(\frac\pi2 \sqrt{\frac{\alpha/2}{1-\alpha\pi/4}}+\pi^{1/6}2^{11/6}\right)D(\nu,\nu)^{1/2}\right\}^{-1}.$$
Next we fix some $\gamma'\in\tilde\cK_C$ and use that
\begin{multline}\label{wappler}
H(\gamma',\tilde\gamma^0)=H(\gamma,\tilde\gamma^0)+H(\gamma',\gamma)\\
+\tr\left[(\gamma'-\gamma)\left(\ln\left(\frac{\half+\gamma}{\half-\gamma}\right)-\ln\left(\frac{\half+\tilde\gamma^0}{\half-\tilde\gamma^0}\right) \right)\right].
\end{multline}
Inserting Eq.~\eqref{SCF_BDF} for our minimizer $\gamma$ and Eq.~\eqref{SCF_free} for $\tilde\gamma^0$, we obtain for any $\gamma'\in\tilde\cK_C$ the formula
\begin{equation}
\cF_T(\gamma')=\cF_T(\gamma)+TH(\gamma',\gamma)+\frac\alpha2 D(\rho_{\gamma'-\gamma},\rho_{\gamma'-\gamma})-\frac\alpha2\iint\frac{\tr_{\C^4}|(\gamma - \tilde\gamma^0)(x,y)|^2}{|x-y|}dx dy.
\end{equation}
We may use one more time \eqref{formula_log} and the self-consistent equation \eqref{SCF_BDF} to obtain
$$TH(\gamma',\gamma)\geq \tr\left(|D_{\gamma}|(\gamma'-\gamma)^2\right).$$
By \eqref{estim_below_scf_operator} and Kato's inequality as before, we eventually get
$$\cF_T(\gamma')\geq \cF_T(\gamma)+\frac\alpha2 D(\rho_{\gamma'-\gamma},\rho_{\gamma'-\gamma})+\left(\frac{2}{\pi d}-\frac\alpha2\right)\iint\frac{\tr_{\C^4}|(\gamma - \tilde\gamma^0)(x,y)|^2}{|x-y|}dx dy.$$
Hence we obtain that any minimizer is unique when $\alpha\pi
d/4\leq1$, as stated. Let us remark that the expression in last term
of \eqref{wappler} is indeed a trace-class operator. It would,
however, have been sufficient to  choose $\gamma'$ as trace class
perturbation of $\gamma^0$ and conclude the rest by a density
argument. \qed

\appendix
\section{Appendix: Integral Representation of Relative Entropy}

Here we shall prove the claim made in Section~\ref{sec:def_H} that (\ref{def_H_log}) and (\ref{def_H_int}) coincide as long as $\gamma-\gamma^0\in \gS_1(\gH_\Lambda)$ and the spectrum of $\gamma^0$ does not contain $\pm 1/2$.

From the integral representation (\ref{formula_log}), we have
$$
a(\ln a - \ln b) = \int_0^\infty \frac{a}{a+t} (a-b) \frac 1{b+t} dt\,.
$$
We split the first factor as
$$
\frac{a}{a+t} = \frac{b}{b+t} +\frac{t}{b+t}(a-b)\frac{1}{a+t}
$$
and obtain
\begin{align}\nonumber
a(\ln a - \ln b) & = \int_0^\infty \frac{t}{b+t}(a-b)\frac{1}{a+t} (a-b) \frac 1{b+t} dt \\ \nonumber &\quad + \int_0^\infty \frac{b}{b+t}  (a-b) \frac 1{b+t} dt\,.
\end{align}
Now if $a-b$ is trace class and the spectrum of $b$ is contained in $(0,\infty)$, then
$$
\tr  \int_0^\infty \frac{b}{b+t}  (a-b) \frac 1{b+t} dt = \tr \int_0^\infty  (a-b) \frac b{(b+t)^2} dt = \tr (a-b)\,.
$$
If we apply this reasoning to $a=1/2+\gamma$, $b=1/2+\gamma^0$ and to $a=1/2-\gamma$ and $b=1/2-\gamma^0$, respectively, we thus obtain
\begin{align}
H(\gamma,\gamma^0)&=\tr\Bigg(\int_0^\ii\frac{t}{\half+\gamma+t}(\gamma-\gamma^0)\frac{1}{\half+\gamma^0+t}(\gamma-\gamma^0)\frac{1}{\half+\gamma^0+t}dt\nonumber\\
&\quad+\int_0^\ii\frac{t}{\half-\gamma+t}(\gamma-\gamma^0)\frac{1}{\half-\gamma^0+t}(\gamma-\gamma^0)\frac{1}{\half-\gamma^0+t}dt\Bigg)\nonumber
\end{align}
Changing variables from $t$ to $u=1/(1+2t)$ in the first
integral and $u = - 1/(1+2t)$ in the second integral, respectively, we arrive at the integral representation (\ref{def_H_int}).

\addcontentsline{toc}{section}{References}
\bibliographystyle{amsplain}

\end{document}